\documentclass[aps,pra,twocolumn,unsortedaddress]{revtex4-1}

\usepackage{amsmath,amsthm,amssymb}

\usepackage{xcolor}

\PassOptionsToPackage{hyphens}{url} 
\usepackage[breaklinks=true]{hyperref}

\newcommand{\R}{{\mathbb R}}

\newcommand{\x}{{\mathbf{r}}}
\newcommand{\rmi}{{\mathrm{i}}}
\newcommand{\rmd}{\,\mathrm{d}}
\newcommand{\jtot}{\mathbf{j}}
\newcommand{\jpara}{\mathbf{j}^{\mathrm{p}}}
\newcommand{\jarb}{\mathbf{k}}
\newcommand\aD {\mathbf{a}}
\newcommand\jparaDm {{\jpara_{\mathrm{m}}}}
\newcommand\FGH {G}
\newcommand\FVR {F_{\mathrm{VR}}}
\newcommand\ED {\mathcal{E}_{\mathrm{D}}}

\newcommand\FD {F_{\mathrm{D}}}
\newcommand\GD {G_{\mathrm{D}}}
\newcommand\Gstat {G_{\mathrm{stat}}}
\newcommand{\A}{\mathbf{A}}

\newcommand\trace[1] {\mathrm{tr}({#1})}

\newcommand\pureqmstate {\psi}
\newcommand\mixedqmstate {\Gamma}

\newcommand\expval[2] {\trace{ {#1} {#2}}}
\newcommand\qmstate {\mixedqmstate}

\newtheorem{theorem}{Theorem}
\newtheorem{proposition}[theorem]{Proposition}
\newtheorem{lemma}[theorem]{Lemma}
\newtheorem{corollary}[theorem]{Corollary}

\DeclareMathOperator*{\stat}{\mathrm{stat}}

\begin{document}




\title{Revisiting density-functional theory of the total current density}


\author{Andre Laestadius}
\email{corresponding author: andre.laestadius@kjemi.uio.no}
\affiliation{Hylleraas Centre for Quantum Molecular Sciences, Department of Chemistry, University of Oslo, P.O.\ Box 1033 Blindern, N-0315 Oslo, Norway}

\author{Markus Penz}
\affiliation{Department of Mathematics, University of Innsbruck, Technikerstra{\ss}e 13/7, A-6020 Innsbruck, Austria}

\author{Erik I. Tellgren}
\email{corresponding author: erik.tellgren@kjemi.uio.no}
\affiliation{Hylleraas Centre for Quantum Molecular Sciences, Department of Chemistry, University of Oslo, P.O.\ Box 1033 Blindern, N-0315 Oslo, Norway}

\date{\today}

\begin{abstract}

Density-functional theory requires an extra variable besides the electron density in order to properly incorporate magnetic-field effects. In a time-dependent setting, the gauge-invariant, total current density takes that role. A peculiar feature of the static ground-state setting is, however, that the gauge-dependent paramagnetic current density appears as the additional variable instead.
An alternative, exact reformulation in terms of the total current density has long been sought but to date a work by Diener is the only available candidate.  In that work, an unorthodox variational principle was used to establish a ground-state 
density-functional theory of the total current density as well as an accompanying Hohenberg--Kohn-like result.
We here reinterpret and clarify Diener's formulation based on a maximin variational principle. Using simple facts about convexity implied by the resulting variational expressions, we prove that Diener's formulation is unfortunately not capable of reproducing the correct ground-state energy and, furthermore, that the suggested construction of a Hohenberg--Kohn map  contains an irreparable mistake. 
\end{abstract}

\pacs{02.30.Jr 02.30.Sa} 
\keywords{Hohenberg--Kohn theorem, current-density-functional theory, magnetic systems, density-functional theory} 

\maketitle

\section{Introduction} 

The Hohenberg--Kohn theorem is commonly regarded as the theoretical foundation of density-functional theory (DFT). 
Omitting technical points~\cite{Lieb1983,Lammert2018,Garrigue2018}, it asserts that the electron density determines the external potential (up to a constant) and therefore the Hamiltonian and all system properties~\cite{Hohenberg1964}.
To include arbitrary magnetic fields into the formalism, DFT needs to be supplemented by an additional basic variable. In current-density-functional theory (CDFT) the paramagnetic current density takes that role~\cite{Vignale1987}. It is also possible to forego any attempt to find a universal functional independent of the external potentials and instead have a formalism that is parametrically dependent on the magnetic field~\cite{GRAYCE}. 
A peculiar feature of CDFT is that it is the paramagnetic current density, and not the gauge-invariant total current density, that enters as a basic variable. This leaves a disconnect between ground-state CDFT and the time-dependent version of the theory, which is naturally formulated using the total current density~\cite{VIGNALE_PRB70_201102,VIGNALE_PRL77_2037}. Additionally, the total current density avoids practical issues arising from having to extract the gauge invariant part of the paramagnetic current density in approximate density functionals~\cite{Trickey,TELLGREN_JCP140_034101,TellgrenSolo2018}. As far as a CDFT for ground states formulated with the total current density is concerned, the question if a Hohenberg--Kohn theorem holds is still open and has attracted some recent attention \cite{Tellgren2012,LaestadiusBenedicks2014,Ruggenthaler2015,Tellgren2018,Garrigue2019,Garrigue2019b}. Several authors have realized that the total current density at best fits awkwardly into standard density-functional approaches and that, in fact, it is incompatible with the \emph{standard} energy minimization principle~\cite{VIGNALE_IJQC113_1422,Tellgren2012,LaestadiusBenedicks2015,TellgrenSolo2018}. However, it has been remarked that energy maximization with respect to the current density is not excluded by any known result~\cite{Tellgren2012} and recent work has shown that the Maxwell--Schr\"odinger energy minimization principle naturally leads to a density-functional theory that features the total current density~\cite{TellgrenSolo2018}. 

To date, a work by Diener~\cite{Diener} is the only candidate for a density-functional theory of the total current that does not modify the underlying Schr\"odinger equation. Logical gaps in his formulation have been identified before~\cite{Tellgren2012,LaestadiusBenedicks2014}, although one specific criticism was mistaken (Proposition~8 in Ref.~\onlinecite{LaestadiusBenedicks2014}, which we correct below at the end of Sec.~\ref{sec:OrthodoxDiener}). Nonetheless, despite the gaps, Diener's unique approach is interesting as it comes tantalizingly close to succeeding and it has so far been unclear whether the approach can be rigorously completed.

In this work, we first clarify the underlying assumptions in Diener's approach by reinterpreting it as based on a maximin variational principle. Based on simple facts about convexity of the resulting energy functional, it can be concluded that Diener's approach is neither capable of reproducing the ground-state energy nor the correct total current density. 
We also establish that Diener's construction of a Hohenberg--Kohn map suffers from an irreparable error: the selection of a vector potential via a stationary search over current densities is not correct.
Our analysis is very general and applies even if previously identified issues~\cite{Tellgren2012,LaestadiusBenedicks2014} could somehow be resolved.

\section{Preliminaries}
\label{sec:Prel}

Our point of departure is the time-independent magnetic Schr\"odinger equation for electrons with the Hamiltonian 
(in SI-based atomic units, compared to Diener~\cite{Diener} we use the convention $e\A \to \A$ and $-e\Phi \to v$ for the potentials)
\begin{equation}
  H(v,\A) = \frac{1}{2} \sum_j \left( -\rmi\nabla_j + \A(\mathbf{r}_j) \right)^2 + \sum_j v(\mathbf{r}_j) + W.
\end{equation}
Here $(v,\mathbf{A})$ are the external electromagnetic potentials and $W = \sum_{i<j} r_{ij}^{-1}$ is the electron--electron repulsion operator. We use the short-hand notation $H_0 = H(0,\mathbf{0})$ for the universal part of the Hamiltonian. Spin has no bearing on the present work and we therefore leave out all spin degrees-of-freedom from the notation.

For pure states $\pureqmstate(\mathbf{r}_1,\ldots,\mathbf{r}_N)$, where $\Gamma = \vert \pureqmstate\rangle\langle \pureqmstate\vert $ is the density matrix, the particle density and paramagnetic current density are given by, respectively, 
\begin{equation}
    \begin{split} \label{eq:dens-para}
  \rho_{\pureqmstate}(\mathbf{r}_1) & = N \int |\pureqmstate|^2 \rmd\mathbf{r}_2 \cdots \rmd\mathbf{r}_N,
  \\
  \jpara_{\pureqmstate}(\mathbf{r}_1) & = N \, \mathrm{Im} \int \bar{\pureqmstate} \nabla_1 \pureqmstate \rmd\mathbf{r}_2 \cdots \rmd\mathbf{r}_N,
  \end{split}
\end{equation}
and with well-known extensions to mixed states. Under a gauge transformation $\A \mapsto \A + \nabla f$, the paramagnetic current density transforms as $\jpara \mapsto \jpara - \rho \nabla f$.  The gauge-invariant, total current density is thus given by $\mathbf{j} = \jpara+\rho\A$.

From a direct calculation, using the densities defined in Eq.~\eqref{eq:dens-para}, 
\begin{equation}\label{eq:known-to-CDFTist}
\begin{split}
    \expval{H(v,\A)}{\qmstate} &= \expval{H_0}{\qmstate} + \int \jpara_\qmstate \cdot \A \rmd \x \\
    &\quad +\int \rho_\qmstate (v+\tfrac{1}{2}\vert \A\vert^2) \rmd \x.
\end{split}
\end{equation}
Using Eq.~\eqref{eq:known-to-CDFTist} the ground-state energy can be obtained from the expression
\begin{equation}
  \begin{split}
    & E(v,\A)  = \inf_{\qmstate} \expval{H(v,\A)}{\qmstate}
            \\
       & = \inf_{\rho,\jpara} \left\{  \FVR(\rho,\jpara) + \int (\rho (v+\tfrac{1}{2} \vert \A \vert^2)  + \jpara\cdot\A) \rmd\mathbf{r}  \right\},
  \end{split}
\end{equation}
where we have introduced the Vignale--Rasolt universal functional \cite{Vignale1987},
\begin{equation}\label{eq:FVR}
    \FVR(\rho,\jpara) = \inf_{\qmstate \mapsto (\rho,\jpara)} \expval{H_0}{\qmstate}.
\end{equation}
A recent result establishes that the infimum in Eq.~\eqref{eq:FVR} can be replaced by a minimum for any physically reasonable densities~\cite{Kvaal2020}. 
It is known that the paramagnetic current density (together with $\rho$) does not determine the external potentials~\cite{Capelle2002}, although the original proof idea~\cite{Vignale1987} can be used to establish a mapping from $(\rho,\jpara)$ to nondegenerate ground states~\cite{Tellgren2012}. This was termed a \emph{weak} Hohenberg--Kohn result in Ref.~\onlinecite{LaestadiusTellgren2018}, where the degenerate case was further analysed.

Another formulation is obtained by introducing the Grayce--Harris semiuniversal density functional \cite{GRAYCE},
\begin{equation}\label{eqGHfun}
  \begin{split}
  &\FGH(\rho,\A)  = \inf_{\qmstate \mapsto \rho} \expval{H(0,\A)}{\qmstate}
                  \\
                  &\quad  =  \int \tfrac{1}{2} \rho \vert \A \vert^2 \rmd\mathbf{r}+ \inf_{\jpara} \left\{ \FVR(\rho,\jpara) +\int \jpara\cdot\A \rmd\mathbf{r}   \right\},
     \end{split}
\end{equation}
which enables the ground-state energy to be written as the (magnetic field-) B-DFT variational principle,
\begin{equation}
    E(v,\A) = \inf_{\rho} \left\{  \FGH(\rho,\A)  + \int \rho  v \rmd\mathbf{r} \right\}.
\end{equation}
The semiuniversal nature of $\FGH(\rho,\A)$ directly leads to a type of Hohenberg--Kohn result: For every fixed $\A$, a positive ground state state $\rho(\mathbf{r}) > 0$ determines $v$ up to a constant~\cite{GRAYCE}.

The relationship between the above two frameworks has recently been highlighted and analyzed~\cite{TellgrenSolo2018,REIMANN_JCTC13_4089}, with particular focus on convexity properties and variational principles connecting the formalisms. [See Appendix~\ref{appConvexStuff} for basic definitions of convexity and related notions.] At least for small vector potentials, the physical interpretation is that convexity of the energy in $\A$ is associated with diamagnetism, while concavity in $\A$ is associated with paramagnetism.  Here, we note that the mixed-state version of $\FVR$ defined in Eq.~\eqref{eq:FVR} is jointly convex in $(\rho,\jpara)$ (but the pure-state version is not, see Proposition~8 in Ref.~\onlinecite{Laestadius2014}). The mixed state version of $\FGH$ is likewise convex in $\rho$; however, it is neither convex nor concave in $\A$. As discussed in Ref.~\onlinecite{TellgrenSolo2018}, the Grayce--Harris functional is paraconcave (``concave up to a square'') in $\A$, i.e., the difference $\bar{\FGH}(\rho,\A) = \FGH(\rho,\A) -  \int \tfrac{1}{2} \rho \vert\A\vert^2 \rmd\mathbf{r}$ is concave. Loosely interpreted in physical terms this means that all systems appear paramagnetic when the diamagnetic term is removed. The corresponding transformation of the ground-state energy $E(v,\A)$ is a change of variables $\bar{E}(u,\A) = E(u-\tfrac{1}{2} \vert\A\vert^2, \A)$, which makes $\bar{E}(u,\A)$ jointly concave in $(u,\A)$, unlike the original $E(v,\A)$.

That $\FGH(\rho,\A)$ cannot be convex in $\A$ is fairly obvious from the physical interpretation. However, since this property will be important in the further results below, we give a full proof.
\begin{proposition}
  For some $\rho$, the Grayce--Harris functional $\FGH(\rho,\A)$ is not convex in $\A$.
\end{proposition}
\begin{proof}
Consider a $\rho$ such that for $\A=0$ one has a ground-state degeneracy that allows for a current $\pm \jpara_{\mathrm{gs}} \neq \mathbf{0}$. 
Both signs are possible for $\jpara_{\mathrm{gs}}$ due to time-reversal symmetry. 
Now, take $\A \neq 0$ such that for one of the ground states one has $\int  \jpara_{\mathrm{gs}} \cdot\A \rmd \x  = -\left| \int  \jpara_{\mathrm{gs}} \cdot \A \rmd\x \right|<0$. 
Then for sufficiently small but nonzero $\A$,
\begin{equation} \label{eq:PenzHawkEye}
  \begin{split}
    &\FGH(\rho,\A)  =  \int  \tfrac{1}{2} \rho |\A|^2 \rmd\mathbf{r}  + \inf_{\qmstate\mapsto\rho} \left\{ \expval{H_0}{\qmstate} +  \int \jpara_{\qmstate}\cdot\A \rmd\mathbf{r}   \right\}
           \\
           & \quad\leq G(\rho,\mathbf 0) + \int \tfrac{1}{2} \rho \vert \A \vert^2 \rmd \x  - \left| \int  \jpara_{\mathrm{gs}}\cdot\A \rmd \x  \right|
           < \FGH(\rho,\mathbf 0).
  \end{split}
\end{equation}
On the other hand, invoking time-reversal symmetry, namely $\FGH(\rho,+\A) = \FGH(\rho,-\A)$, the assumption of convexity of $\FGH(\rho,\A)$ in $\A$ would have entailed $G(\rho,\A) = \tfrac{1}{2} (\FGH(\rho,\A) + \FGH(\rho,-\A) ) \geq G(\rho,\mathbf 0)$, in contradiction with Eq.~\eqref{eq:PenzHawkEye}.
\end{proof}

Note that the above result substantially understates the extent of the non-convexity---it is not restricted at all to very special densities $\rho$. For example, some $\rho$ correspond to paramagnetic systems that have concave $G(\rho,\A)$ in $\A$. Moreover, most $\rho$ are such that increasing the magnetic-field strength will reorder the energy spectrum so that states with permanent paramagnetic currents eventually become the ground state. These level crossings introduce non-convexity as well.

\section{Diener's formulation as a maximin variational principle}

Next, we turn to Diener's unconventional attempt to formulate a total current-density-functional theory. Diener's formalism is greatly simplified and clarified by starting from the ground-state energy and algebraically manipulating the formula until we obtain a variational expression that can be related to his working equations. Taking the B-DFT variational principle as the point of departure, it is indeed sufficient to rewrite the Grayce--Harris functional. Letting $\jarb$ denote an arbitrary current density, we begin by adding an energy term that clearly gives a vanishing net contribution:
\begin{widetext}
\begin{equation}
    \label{eqDinminimax}
  \begin{split}
    G(\rho,\A) & =  \int \tfrac{1}{2} \rho \vert\A \vert^2 \rmd\mathbf{r} + \inf_{\jpara} \left\{  \FVR(\rho,\jpara) +\int \jpara\cdot\A \rmd\mathbf{r}  - \inf_{\jarb} \int \frac{|\jpara + \rho \A - \jarb|^2}{2\rho} \rmd\mathbf{r} \right\}
    \\
    & = \inf_{\jpara} \sup_{\jarb} \left\{ \FVR(\rho,\jpara)  + \int \jarb\cdot\A \rmd\mathbf{r} - \int \frac{|\jpara-\jarb|^2}{2\rho} \rmd\mathbf{r} \right\}  .
  \end{split}
\end{equation}
\end{widetext}
While $\jarb$ is a dummy variable that is being optimized over, its value at the solution to the above minimax problem will satisfy $\mathbf{k} = \jpara + \rho \A$ and hence exactly reproduce the total current density. This way, the issue that the correct energy cannot be obtained from a \emph{standard} minimization principle for the total current density is avoided. Using the general fact that $\inf_x \sup_y f(x,y) \geq \sup_y \inf_x f(x,y)$, we next obtain
\begin{equation}
  \label{eqDinmaximin}
  \begin{split}
    & G(\rho,\A) \\
    &\geq \sup_{\jarb} \inf_{\jpara} \left\{\FVR(\rho,\jpara) + \int \jarb\cdot\A \rmd\mathbf{r}  - \int \frac{|\jpara-\jarb|^2}{2\rho} \rmd\mathbf{r} \right\}
    \\
    & = \sup_{\jarb} \left\{ \FD(\rho,\jarb) + \int \jarb\cdot\A \rmd\mathbf{r}  \right\}
     =: \GD(\rho,\A).
  \end{split}
\end{equation}
We have above introduced $\GD(\rho,\A)$ and identified Diener's proposed total current-density functional 
\begin{equation}
  \label{eqDienerFun}
  \begin{split}
    \FD(\rho,\jarb) & = \inf_{\jpara}  \left\{ \FVR(\rho,\jpara) - \int \frac{|\jpara-\jarb|^2}{2\rho} \rmd\mathbf{r} \right\}
    \\
      & = \inf_{\qmstate\mapsto \rho} \left\{ \expval{H_0}{\qmstate} - \int \frac{|\jpara_{\qmstate}-\jarb|^2}{2\rho} \rmd\mathbf{r} \right\}.
  \end{split}
\end{equation}

The issue now arises as to whether the above maximin principle always achieves equality in Eq.~\eqref{eqDinmaximin}. If this were true, we would have succeeded in expressing the ground-state energy in terms a universal functional $\FD$ of the total current density. Unfortunately, this can immediately be disproven on the basis of convexity properties: The right-hand side of Eq.~\eqref{eqDinmaximin}, i.e., $\GD$, is manifestly convex in $\A$, and hence can only describe diamagnetic systems, whereas the Grayce--Harris functional $\FGH(\rho,\A)$ is nonconvex in $\A$. This establishes the following result:
\begin{proposition}
  \label{PropGDneqFGH}
  For some $(\rho,\A)$, we have a strict inequality $\FGH(\rho,\A) > \GD(\rho,\A)$.
\end{proposition}

A remaining issue is whether Diener's functional $\FD(\rho,\jarb)$ or the variational principle for $\GD(\rho,\A)$ are useful for other purposes, such as reconstructing the correct external vector potential from an input pair $(\rho,\jtot = \jpara + \rho\A)$ or delivering the correct total current density from a pair $(\rho,\A)$. The former would establish a Hohenberg--Kohn-type mapping, since then $(\rho,\jtot)$ determines $(\rho,\A)$ up to a gauge. In a next step one could use the B-DFT extension of the Hohenberg--Kohn theorem to determine $v$~\cite{GRAYCE,LaestadiusBenedicksPenz,Garrigue2019}. In Diener's work, this is in fact the primary intended use of the minimization principle that defines $\FD$. Moreover, he relies heavily on the fact that a state $\qmstate$ and an arbitrary vector field $\jarb$ can be ``related'' through the effective vector potential
\begin{equation} \label{eq:aD}
  \aD(\qmstate,\jarb) := \frac{\jarb - \jpara_{\qmstate}}{\rho_{\qmstate}}.
\end{equation}
By definition we have $\jarb = \jpara_{\qmstate} + \rho_{\qmstate} \aD(\qmstate,\jarb)$, mimicking the relationship between the total current density, the paramagnetic current, and the actual external vector potential. If supplying the true total current density $\jtot = \jpara + \rho \A$ to $\FD(\rho,\jtot)$ always yields a minimizer $\qmstate_{\mathrm{m}}$ in Eq.~\eqref{eqDienerFun} such that $\aD(\qmstate_{\mathrm{m}},\jarb) = \A$, a Hohenberg--Kohn-type mapping would be established. More precisely, since the input to $\FD$ is gauge invariant, the external vector potential can at best be determined up to a gauge. Hence, we have to allow for $\aD(\qmstate_{\mathrm{m}},\jarb) = \A + \nabla f$ and multiple gauge dependent minimizers $\jparaDm$ in Eq.~\eqref{eqDienerFun}, one of which corresponds to a gauge in which $\aD(\qmstate_{\mathrm{m}},\jarb) = \A$. This weaker statement would be sufficient to establish the Hohenberg--Kohn-type mapping. Unfortunately, the next proposition shows that such an $\FD$-based mapping does not exist. 
\begin{proposition} \label{prop:EITmasterwork}
  For some $(\rho,\A)$, Diener's current density functional $\FD$ fails to reconstruct the external potential. That is, for any minimizer $\jparaDm$ in Eq.~\eqref{eqDienerFun} we have
  \begin{equation}
        \frac{\jtot - \jparaDm}{\rho} \neq \A.
  \end{equation}
\end{proposition}
\begin{proof} Fix an arbitrary pair $(\rho,\A)$ for which there exist current densities $(\jpara_0,\jtot_0=\jpara_0 + \rho\A)$ that solve the minimax problem Eq.~\eqref{eqDinminimax}. Inserting $\jtot_0$ into Diener's functional yields
\begin{equation}
  \begin{split}
    \FD(\rho,\jtot_0) & = \inf_{\jpara}  \left\{ \FVR(\rho,\jpara) - \int \frac{|\jpara-\jtot_0|^2}{2\rho} \rmd\mathbf{r} \right\}
    \\
    & =\FVR(\rho,\jparaDm) - \int \frac{|\jparaDm-\jtot_0|^2}{2\rho} \rmd\mathbf{r},
  \end{split}
\end{equation}
where $\jparaDm$ is a minimizer. Now assume, arguendo, that this minimizer can always be chosen to satisfy
\begin{equation}
  \label{eqDienerAlwaysA}
  \frac{\jtot_0 - \jparaDm}{\rho} = \A.
\end{equation}
But this is equivalent to $\jparaDm = \jtot_0 - \rho \A = \jpara_0$. As a direct consequence, we have the lower bound
\begin{equation}
  \begin{split}
    &\GD(\rho,\A)  = \sup_{\jarb} \left\{ \FD(\rho,\jarb) + \int \jarb\cdot\A \rmd\mathbf{r} \right\}
    \\
    & \geq \FD(\rho,\jtot_0) + \int \jtot_0\cdot\mathbf{A} \rmd\mathbf{r} 
    \\
    & =  \FVR(\rho,\jpara_0) + \int \jtot_0\cdot\mathbf{A} \rmd\mathbf{r}  - \int \frac{|\jpara_0-\jtot_0|^2}{2\rho} \rmd\mathbf{r}
    \\
    & = \FVR(\rho,\jpara_0) + \int (\jpara_0+\rho\A)\cdot\mathbf{A} \rmd\mathbf{r}  - \int \frac{|\rho\A|^2}{2\rho} \rmd\mathbf{r} \\
    & = \FGH(\rho,\A).
   \end{split}
\end{equation}
Combining the above bound with the fact that $\GD(\rho,\A) \leq \FGH(\rho,\A)$ from Eq.~\eqref{eqDinmaximin}, we have established that $\GD(\rho,\A) = \FGH(\rho,\A)$ for arbitrary $(\rho,\A)$. This, however, is impossible in light of  Proposition~\ref{PropGDneqFGH}. Hence, we conclude that the assumption that the minimizer $\jparaDm$ can always be chosen to satisfy Eq.~\eqref{eqDienerAlwaysA} is false, which completes the proof.
\end{proof}

It should be noted that we do not need to explicitly impose that the total current density arising from an eigenstate is divergence-free. This condition, $\nabla\cdot\jarb=0$, is not needed in the minimax principle for $\FGH$. However, it possibly makes a difference in the maximin principle $\GD$, yet adding it does not circumvent the problems noted above.

Finally, it must be remarked that in his original work, Diener actually relies on a stationarity principle for a quantity $G_{\mathrm{stat}}$, rather than on the above maximin principle for $\GD$. However, this difference is inessential and, in fact, only adds to the problems identified above. The following bounds are immediate:
\begin{widetext}
\begin{equation}
   \begin{split}
    \FGH(\rho,\A) & = \inf_{\jpara} \sup_{\jarb} \left\{  \FVR(\rho,\jpara) + \int \jarb\cdot\A \rmd\mathbf{r} - \int \frac{|\jpara-\jarb|^2}{2\rho} \rmd\mathbf{r} \right\}
          \\
          & \geq
    \GD(\rho,\A)  = \sup_{\jarb} \inf_{\jpara} \left\{\FVR(\rho,\jpara) + \int \jarb\cdot\A \rmd\mathbf{r}  - \int \frac{|\jpara-\jarb|^2}{2\rho} \rmd\mathbf{r} \right\}
          \\
          & \geq
    G_{\mathrm{stat}}(\rho,\A)  = \stat_{\jarb} \inf_{\jpara} \left\{ \FVR(\rho,\jpara)  + \int \jarb\cdot\A \rmd\mathbf{r} - \int \frac{|\jpara-\jarb|^2}{2\rho} \rmd\mathbf{r} \right\}.
   \end{split}
\end{equation}
\end{widetext}
Hence, by Proposition~\ref{PropGDneqFGH} it follows that $\FGH(\rho,\A) > \GD(\rho,\A) \geq G_{\mathrm{stat}}(\rho,\A)$, for some $(\rho,\A)$. The problems with the maximin principle for $\GD$ thus directly carry over to the stat-min principle for $G_{\mathrm{stat}}$. Naturally, a pure minimization principle, obtained by replacing the maximization over $\jarb$ by a minimization, can only make problems worse.

\section{Diener's original formulation}
\label{sec:OrthodoxDiener}

The previous section reinterpreted Diener's formulation in terms of a maximin principle. In the present section, we provide a direct disproof in terms of Diener's original concepts. As was already clear from Proposition~\ref{prop:EITmasterwork}, Diener's proof is unfortunately in error and $\FD$ cannot be used for a Hohenberg--Kohn result in CDFT.

Recall Eq.~\eqref{eq:aD}, where for given $\qmstate$ and $\jarb$ we have the vector potential $\aD(\qmstate,\jarb) = (\jarb - \jpara_{\qmstate})/\rho_{\qmstate}$. The following proposition is a direct consequence of Eq.~\eqref{eq:known-to-CDFTist}.  

\begin{proposition}[Eq.~(6) in Diener~\cite{Diener}] \label{prop:Eq6D}
Let $H(v,\A)$ be fixed. Then for any $\qmstate$ and any current density $\jarb$
\begin{equation}  \label{AL:eq1} 
 \begin{split}
	 \expval{H(v,\A)}{\qmstate} = \ED(\qmstate,\jarb)   + \int (   \jarb \cdot \A + \rho_{\qmstate} v )\rmd\x & \\
	  +  \int \tfrac{1}{2} \rho_{\qmstate} \vert \A- \aD(\qmstate,\jarb) \vert^2 \rmd\x &,
\end{split}
\end{equation}
with 
\begin{align}
	 \ED(\qmstate,\jarb)  &= \expval{H_0}{\qmstate} - \int \frac{ \vert \jarb - \jpara_{\qmstate} \vert^2}{2\rho_{\qmstate}}\rmd\x .
 \end{align}
\label{prop:D6}
\end{proposition}

Note that $\ED(\qmstate,\jarb) = \expval{H_{\aD(\qmstate,\jarb)}}{\qmstate}$ with $H_{\aD(\qmstate,\jarb)} = H_0 - \tfrac{1}{2} \sum_j \vert \aD(\qmstate,\jarb; \mathbf r_j) \vert^2$, i.e., $\ED$ can be viewed as an expectation value over a state-dependent Hamiltonian $H_{\aD(\qmstate,\jarb)}$. Equation~\eqref{AL:eq1} can also be stated as 
 \begin{equation} \label{eq:G1}
 \begin{split}
	 \ED(\qmstate,\jarb) + \int(   \jarb\cdot \A +  \rho_{\qmstate} v) \rmd \x = \expval{H(v,\A)}{\qmstate}& \\
 -  \int \tfrac{1}{2} \rho_{\qmstate} \vert \A- \mathbf a(\qmstate,\jarb) \vert^2  \rmd \x &.
\end{split}
 \end{equation}
On the left-hand side of Eq.~\eqref{eq:G1} we have $\ED$ (albeit not a functional of the densities) and a \emph{linear} coupling between $(v,\A)$ and the variables $(\rho,\jarb)$. This mimics the situation that one has in density-only DFT, however, with one important difference: the expression on the left-hand side of Eq.~\eqref{eq:G1} does not equal the expectation value $\expval{H(v,\A)}{\qmstate}$. 
To obtain a density-functional setting, Diener minimized the left-hand side of Eq.~\eqref{eq:G1} over all $\qmstate \mapsto \rho$ and 
transitioned from $\ED(\qmstate,\jarb)$ to a density functional by defining 
\begin{align}
	 \FD(\rho,\jarb) := \inf_{\qmstate \mapsto \rho} \ED(\qmstate,\jarb),
\end{align}
which is equivalent to Eq.~\eqref{eqDienerFun}.
The existence and uniqueness of minimizers of $\FD$ was never investigated by Diener---a possible minimum was simply tacitly assumed. As far as the attempt to obtain a Hohenberg--Kohn theorem is concerned, Diener's proof cannot be completed, as will be demonstrated here based on Proposition~\ref{prop:EITmasterwork}. 
However, we will first make an attempt at providing the best possible presentation of Diener's argument.

A word on notation: If a minimizer of $\FD(\rho,\jarb)$ exists we denote it by $\qmstate_{\mathrm{m}}$ and call it a ``Diener minimizer''. For such a $\qmstate_\mathrm{m}$ we have  $\qmstate_{\mathrm{m}} \mapsto \rho$ and
\begin{equation}
 \FD(\rho,\jarb) = \ED( \qmstate_{\mathrm{m}}, \jarb ).
\end{equation}
That such a minimizer indeed can be guaranteed to exist under certain assumptions is proven in Appendix~\ref{app:proof}. 

Diener has formulated an unorthodox variational principle, Eqs.~(10) and (15) in Ref.~\onlinecite{Diener}, which we restate in the following proposition.
 
\begin{proposition}[Diener's generalized variational principle] \label{prop:Dvarp}
Let $v,\A$ be fixed. 
Diener's functional $\FD$ verifies for any $\rho$, for any $\qmstate\mapsto \rho$, and any 
current density $\jarb$, the inequality
\begin{equation}\label{in1}
\begin{split}
\FD(\rho,\jarb)  +  \int (   \jarb \cdot \A  +  \rho v)\rmd \x   \leq \expval{H(v,\A)}{\qmstate} &   \\
-  \int \tfrac{1}{2} \rho \vert \A- \mathbf{a}(\qmstate,\jarb)\vert^2 \rmd \x  &.     
\end{split}
\end{equation}
Moreover, suppose a minimizer $\qmstate_{\mathrm{m}}$ of $\FD(\rho,\jarb)$ exists, then  
\begin{equation} \label{in1-1}
    \begin{split}
            & \expval{H(v,\A)}{\qmstate_{\mathrm{m}}}   -  \int \tfrac{1}{2} \rho \vert \A- \aD(\qmstate_{\mathrm{m}},\jarb) \vert^2 \rmd\x \\
    & \quad \leq \expval{H(v,\A)}{\qmstate}  -  \int  \tfrac{1}{2} \rho \vert \A- \aD(\qmstate,\jarb) \vert^2 \rmd\x.
    \end{split}
\end{equation}
\end{proposition}

\begin{proof}
The first inequality, Eq.~\eqref{in1}, follows from minimizing the left-hand side of Eq.~\eqref{eq:G1} over $\qmstate \mapsto \rho$,
 \begin{equation} \label{eq:AL-add-on1}
     \begin{split}
         &\FD(\rho,\jarb) + \int (  \mathbf  \jarb \cdot \A  +  \rho v) \rmd\x \\
        &= \inf_{\qmstate\mapsto \rho} \left\{
	 \ED(\qmstate,\jarb) + \int (  \jarb \cdot \A  +  \rho v) \rmd \x \right\} \\ 
	 &\leq \trace{ H(v,\A)\qmstate } 
	 -  \int \tfrac 1 2 \rho_\qmstate \vert \A - \mathbf a(\qmstate,\jarb) \vert^2 \rmd\x  . 
     \end{split}
 \end{equation}

To obtain Eq.~\eqref{in1-1}, we note that the left-hand side of Eq.~\eqref{eq:AL-add-on1} can be rearranged into
 \begin{equation}
     \begin{split}
     &\ED(\qmstate_\mathrm{m},\jarb) + \int (      \jarb \cdot \A +  \rho v ) \rmd \x \\
     &= \trace{ H_0 \qmstate_\mathrm{m} }
     + \int (   \jarb \cdot \A +  \rho v) \rmd \x 
     -  \int  \tfrac 1 2 \rho \vert \aD(\qmstate_{\mathrm{m}},\jarb) \vert^2 \rmd\x \\
     &= \trace{ H(v, \A) \qmstate_\mathrm{m} }  -  \int \tfrac{1}{2} \rho \vert \A- \aD(\qmstate_{\mathrm{m}},\jarb) \vert^2 \rmd\x .    
     \end{split}
 \end{equation}
Here we used $\FD(\rho,\jarb) =\ED(\qmstate_\mathrm{m},\jarb)$ and Eq.~\eqref{eq:known-to-CDFTist} with $\jarb = \jpara_{\qmstate_\mathrm{m}} + \rho\aD(\qmstate_{\mathrm{m}},\jarb)$.
\end{proof}

In light of Proposition~\ref{prop:Dvarp}, a natural question to ask is the relation between Diener minimizers $\qmstate_\mathrm{m}$ and ground states $\qmstate_0$. We can offer the following answer: 
For a Hamiltonian $H(v,\A)$ where $\A=\aD(\qmstate_{\mathrm{m}},\jarb)$, the Diener minimizer $\qmstate_\mathrm{m}$ is a ground state. 
However, a ground state $\qmstate_0\mapsto (\rho_0,\jtot_0)$ generally does not need to be a minimizer of $\FD(\rho_0,\jtot_0)$ (the proof is given below in Proposition~\ref{prop:amazinEIT}).

\begin{corollary} \label{cor:gs-thing}
Suppose $\qmstate_\mathrm{m} \mapsto \rho_0$ to be a Diener minimizer for $\FD(\rho_0,\jarb)$ and that $H(v,\A)$, with $\A = \aD(\qmstate_{\mathrm{m}},\jarb)$, has $\rho_0$ as a ground-state density. Then $v$ is unique up to a constant. Moreover, for any $\Gamma \mapsto \rho_0$ it holds
(Eq.~(15) in Diener)
\begin{equation} \label{eq:D-var-prin}
\begin{split}
        \expval{H(v,\A)}{\qmstate_{\mathrm{m}}}  \leq \expval{H(v,\A)}{\qmstate}   -  \int \tfrac{1}{2} \rho \vert \A- \aD(\qmstate,\jarb) \vert^2 \rmd\x,
\end{split}
\end{equation}
and as a consequence $\qmstate_\mathrm{m}$ is a ground state for $H(v,\A)$ and $\jarb =  \jpara_{\qmstate_\mathrm{m}} + \rho_0 \A =:\jtot_0$ is the ground-state current density. Further, $\aD(\qmstate_0,\jtot_0) = \aD(\qmstate_{\mathrm{m}},\jtot_0)$ for any other ground state of $H(v,\A)$ with $\qmstate_0\mapsto \rho_0$. 
\end{corollary}

Note that Corollary~\ref{cor:gs-thing} implies that all ground states with $\qmstate_0\mapsto \rho_0$ have the same paramagnetic current densities, since the diamagnetic part always is $\rho_0 \A$. [See also the joint-degeneracy theorem in Ref.~\onlinecite{Capelle2007}.]

\begin{proof}[Proof of Corollary~\ref{cor:gs-thing}]
Firstly, since the vector potential $\A=\aD(\qmstate_{\mathrm{m}},\jarb)$ in $H(v,\A)$ is fixed and $\rho_0$ is by assumption a ground-state density, the Hohenberg--Kohn result of B-DFT~\cite{GRAYCE} gives that $v$ is determined up to a constant. 

The inequality in Eq.~\eqref{eq:D-var-prin} is a direct consequence of Eq.~\eqref{in1-1} with $\aD(\qmstate_{\mathrm{m}},\jarb) = \A$.
Since Eq.~\eqref{eq:D-var-prin} implies the weaker bound $\expval{H(v,\A)}{\qmstate_{\mathrm{m}}} \leq \expval{H(v,\A)}{\qmstate}$
for any $\Gamma\mapsto \rho_0$, it follows that $\qmstate_{\mathrm{m}}$ is a ground state of $H(v,\A)$.
In particular,  Eq.~\eqref{eq:D-var-prin} gives for any ground state of $H(v,\A)$ with $\Gamma_0\mapsto \rho_0$, 
\begin{equation}
\int \rho_0 \vert \A-\mathbf a(\qmstate_0,\jarb)\vert^2  \rmd \x = 0
\end{equation}
and thus
$\rho_0\vert \A- \mathbf a(\qmstate_0,\jarb) \vert^2 =0$ almost everywhere (a.e.). 
By the unique-continuation property from sets of positive measure~\cite{Garrigue2019,LaestadiusBenedicksPenz}, we have $\vert \{\rho_0=0 \} \vert=0$, so $\mathbf a(\qmstate_0,\jarb) =  \A = \aD(\qmstate_{\mathrm{m}},\jarb)$ (a.e.). 
\end{proof}

The main question at this point is, how to guarantee the required $\A=\aD(\qmstate_{\mathrm{m}},\jarb)$. To meet that end, Diener suggested in Ref.~\onlinecite{Diener} to 
choose $\A$ and, for arbitrary $\rho$, find the stationary point, 
\begin{equation}
    \jtot_\mathrm{stat}(\rho,\A) = \mathrm{arg} \, \stat_{\jarb} \left\{  F_\mathrm{D}(\rho,\jarb) + \int  \jarb\cdot \A \rmd\x  \right\},
\end{equation}
where it was implicitly assumed that (i) there is a Diener minimizer $\qmstate_{\mathrm{m}}$ and (ii) $\FD(\rho,\jarb)$ is differentiable with respect to $\jarb$. Under these  assumptions, Diener claimed that this $\jtot_{\mathrm{stat}}$ has the desired property
\begin{equation}
    \aD(\qmstate_\mathrm{m}, \jtot_{\mathrm{stat}}) =\A \quad \text{(up to a gauge)}.
\end{equation}
For a ground state $\qmstate_0 \mapsto \rho_0$ of $H(v,\A)$ one then has by Eq.~\eqref{eq:G1}
\begin{equation} \label{eq:AL91}
    \begin{split}
        \FD(\rho_0,\jtot_\mathrm{stat}(\rho_0,\A)) + \int (  \jtot_{\mathrm{stat}}(\rho_0,\A) \cdot \A+  \rho_0 v)\rmd \x &\\
        \leq  \expval{H(v,\A)}{\qmstate_0},
    \end{split}
\end{equation}
where the left-hand side equals $\expval{H(v,\A)}{\qmstate_{\mathrm{m}}}$.
To sum it up, a stationary variation over $\jarb$ is thought to select the correct $\aD(\qmstate_\mathrm{m}, \jtot_{\mathrm{stat}}) =\A$ while in the next step minimizing over densities gives the ground-state energy because of Eq.~\eqref{eq:AL91},
\begin{equation}
    \begin{split}
       E(v,\A) 
       &= \inf_\rho \stat_\jarb \left\{ \FD(\rho,\jarb) + \int ( \jarb\cdot\A  + \rho v ) \rmd\x \right\}.
    \end{split}
\end{equation}

The attempted proof of Diener for a Hohenberg--Kohn result then relies on the augmented variational principle in Eq.~\eqref{eq:D-var-prin} that moreover has to be a strict inequality for $\qmstate$ not being a Diener minimizer $\qmstate_{\mathrm{m}}$. But since Corollary~\ref{cor:gs-thing} shows that under certain assumptions such Diener minimizers are ground states, an \emph{additional} condition of uniqueness of ground states gives a strict inequality. The usual Hohenberg--Kohn argument by contradiction could then be completed by means of Eq.~\eqref{eq:D-var-prin}.

Furthermore, under the assumption that $\aD(\qmstate_\mathrm{m}, \jtot_{\mathrm{stat}})=\A$ gets selected, there is also a more direct argument available. 
Suppose that $(\rho_0,\jtot_0)$ is the ground-state density pair of two different Hamiltonians with vector potentials $\A$ and $\A'$, respectively. Then if for $\FD(\rho_0,\jarb) + \int  \jarb\cdot \A \rmd \x$ and 
$\FD(\rho_0,\jarb) + \int \jarb\cdot \A' \rmd\x$ Diener's stationary search selects 
$\jtot_0 = \jtot_\mathrm{stat}(\rho_0,\A) = \jtot_\mathrm{stat}(\rho_0,\A')$ that has $\aD(\qmstate_\mathrm{m}, \jtot_0)$ equal to both $\A$ and $\A'$ up to a gauge, 
then the magnetic field is the same for both systems. 
The Hohenberg--Kohn result, i.e., that the scalar potentials also are equal (up to an additive constant), then would follow by the B-DFT result of Grayce and Harris~\cite{GRAYCE}. 

Alas, as a corollary to our main Proposition~\ref{prop:EITmasterwork}, the next proposition shows that Diener's stationary search (as suggested and erroneously proved in Ref.~\onlinecite{Diener}) does \emph{not} select $\aD(\qmstate_\mathrm{m}, \jtot_{\mathrm{stat}})=\A$ up to a gauge. Furthermore, we also have, as a corollary to Proposition~\ref{prop:EITmasterwork}, that ground states are \emph{not} in general minimizers of the Diener functional $\FD$.

\begin{proposition} \label{prop:amazinEIT}
(i) Let $\rho$ and $\A$ be fixed. The Diener optimization
\begin{equation}
\Gstat(\rho,\A) = \stat_\jarb \left\{  F_\mathrm{D}(\rho,\jarb) + \int  \jarb\cdot\A \rmd\x \right\}
\end{equation}
does not in general select $\jtot_\mathrm{stat}$ such that $\mathbf{a}(\qmstate_\mathrm{m}, \jtot_{\mathrm{stat}}) = \A$ (up to a gauge).

(ii) A ground state with the density pair $(\rho,\jtot)$ is not in general a Diener minimizer of $\FD(\rho,\jtot)$.

\end{proposition}
\begin{proof}

For (i), we shall establish $\GD  (\rho,\A) \geq \FGH(\rho,\A)$ for arbitrary $(\rho,\A)$, which by 
Proposition~\ref{PropGDneqFGH} is a contradiction. 
If no stationary point exists there is nothing to prove. Therefore assume that the value of $\Gstat(\rho,\A)$ is realized by a $\jtot_{\mathrm{stat}}$ and its contribution via $\FD(\rho,\jtot_{\mathrm{stat}})$, which, in turn, is realized by some state $\qmstate_{\mathrm{m}}$ with paramagnetic current density $\jparaDm$. Because the input to $\FD$ is gauge invariant, any gauge transformed state $\qmstate'_{\mathrm{m}}$, with $\jparaDm' = \jparaDm + \rho \nabla\chi$ and gauge function $\chi$, is an equally valid minimizer. By stipulation, we have $\rho \, \aD(\qmstate_{\mathrm{m}},\jtot_{\mathrm{stat}}) = \jtot_{\mathrm{stat}} - \jparaDm = \rho (\A + \nabla f)$, where $f$ is a gauge function for $\A$. It follows that choosing $\chi = f$ reproduces the external vector potential exactly, i.e., $\aD(\qmstate'_{\mathrm{m}},\jtot_{\mathrm{stat}}) = \A$, and also $\jtot_{\mathrm{stat}} = \jparaDm' + \rho\A$. Hence,
\begin{equation}
 \begin{split}
\GD & (\rho,\A)  \geq \Gstat(\rho,\A) = \stat_{\jarb}  \left\{ F_\mathrm{D}(\rho,\jarb) +  \int   \jarb\cdot \A \rmd\mathbf{r}   \right\} \\
 & =  \expval{H_0}{\qmstate'_{\mathrm{m}}} + \int  \jtot_{\mathrm{stat}}\cdot \A \rmd\x -  \int \tfrac{1}{2} \rho \vert \A \vert^2 \rmd\x 
        \\
 & = \expval{H_0}{\qmstate'_{\mathrm{m}}} + \int \left( \jparaDm'\cdot\A + \tfrac{1}{2} \rho |\A|^2 \right) \rmd\mathbf{r} \geq \FGH(\rho,\A).
 \end{split}
\end{equation}

For part (ii), we demonstrate that the assumption that a ground state also is a Diener minimizer leads to a contradiction. Let $\A(\mathbf{r}) = \frac{1}{2} B \mathbf{e}_z\times\mathbf{r}$ be a vector potential representing a uniform magnetic field along the $z$-axis. Let $v_B(\mathbf{r}) = -Z/|\mathbf{r}| + \frac{1}{2} (\omega_0^2 - \tfrac{1}{4} B^2)(x^2+y^2)$, with $\omega_0\neq 0$ and note that the effective scalar potential $u = v_B + \tfrac{1}{2} |\A|^2 = -Z/|\mathbf{r}| + \frac{1}{2} \omega_0^2 (x^2+y^2)$ is independent of $B$. The Hamiltonian $H(v_B,\A) = H(v_0,\mathbf{0}) + \frac{1}{2} B L_z$ has cylindrical symmetry and the eigenstates therefore have quantized angular momentum component $L_z = -\rmi\sum_j [\mathbf{r}_j\times\nabla_j]_z$. Due to the quantization, the paramagnetic term becomes a trivial shift and the ground state is piecewise constant as a function of $B$, with jumps corresponding to level crossings. Consequently, the ground-state density $\rho$ is also piecewise constant in $B$. For values of $Z$ that correspond to an open-shell atom (e.g., a carbon atom with $Z=6$ and six electrons), there is a ground state $\qmstate_{-M}$ for $B\geq 0$ with $\expval{L_z}{\qmstate_{-M}} = -M$. For $B\leq 0$, the ground state is $\qmstate_{+M} = \qmstate_{-M}^*$ with the same density, $\qmstate_{\pm M} \mapsto \rho$, but $\expval{L_z}{\qmstate_{+M}} = +M$. For sufficiently small $|B|$, the Grayce--Harris functional is given by
\begin{equation}
    \FGH(\rho,\A) = \FGH(\rho,\mathbf{0}) - \frac{1}{2} |MB| + \int \tfrac{1}{2} \rho \vert \A \vert^2 \rmd\mathbf{r},
\end{equation}
which is nonconvex in $B$ because of the term $-\tfrac{1}{2}|MB|$ and is independent of the sign of $B$.

For $B>0$, the total current density is given by $\jtot_{+} = \jpara_{\qmstate_{-M}} + \rho \A$ and for $B<0$ it is $\jtot_{-} = \jpara_{\qmstate_{+M}} + \rho \A = -\jtot_{+}$. By stipulation, the ground states $\qmstate_{\pm M}$ for all sufficiently small $|B|$ are also a minimizers of $\FD(\rho,\jtot_{\pm})$. Then
\begin{equation}
  \begin{split}
    \GD(\rho,\A) & \geq \FD(\rho, \jtot_{\pm}) + \int \jtot_{\pm}\cdot\A \rmd \mathbf{r}
          \\
    & = \expval{H_0}{\qmstate_{\mp M}} - \int (\tfrac{1}{2} \rho \vert \A \vert^2 + \jtot_{\pm}\cdot\A) \rmd \mathbf{r}
       \\
     & = \FGH(\rho,\A).
  \end{split}
\end{equation}
Combined with the generic fact $\FGH(\rho,\A) \geq \GD(\rho,\A)$, we now have $\FGH(\rho,\A) = \GD(\rho,\A)$ for a whole interval of small $|B|$. However, $\GD(\rho,\A)$ is convex in $\A$ (and therefore also in $B$) and therefore cannot equal the nonconvex $\FGH(\rho,\A)$ on an interval of small $|B|$. This contradiction completes the proof.
\end{proof}

In the previous section, Propositions~\ref{PropGDneqFGH} and~\ref{prop:EITmasterwork} established via a reinterpretation as a minimax principle that Diener's approach cannot work, since it attempts to derive claims that are false. Proposition~\ref{prop:amazinEIT} above further shows, in the terminology of Ref.~\onlinecite{Diener}, that the central steps in Diener's reasoning towards a Hohenberg--Kohn-like result fail. Our results here are thus definitive and go further than previous critiques, which identified an unfounded strict inequality, a self-consistency condition that would require further analysis, and a variational collapse for specific types of ground state densities~\cite{Tellgren2012,Laestadius2014}.

It may also be instructive to note a case where Diener's approach does go through---albeit under extreme restrictions. When only $\qmstate$ with vanishing $\jpara_\qmstate$ are allowed (or when only real valued states are allowed), we can choose a $H(v,\A)$ with ground state $\Gamma_0$ and densities $\rho_0$ and $\jtot_0 = \mathbf{0}+\rho_0 \A$. In this case, the densities trivially determine the external vector potential $\A=\jtot_0/\rho_0$. However, the correct $\jarb=\jtot_0=\rho_0\A$ is also recovered from the variational principle
\begin{equation} \label{eq:37}
\begin{split}
&\Gstat(\rho_0,\A)= \stat_{\jarb} \left\{ \FD(\rho_0,\jarb) + \int  \jarb\cdot \A \rmd\x  \right\} \\
&  = \stat_{\jarb} \left\{ \inf_{\qmstate\mapsto \rho_0,\jpara_\qmstate=\mathbf{0}} \expval{H_0}{\qmstate} -  \int \frac{\vert \jarb \vert^2}{2 \rho_0} \rmd \x + \int  \jarb\cdot \A \rmd\x     \right\}.
\end{split}
\end{equation}
This restrictive case works because the coupling term $\jarb\cdot\jpara_\qmstate/\rho_0$ is absent and $\inf_{\qmstate\mapsto \rho_0,\jpara_\qmstate=\mathbf{0}} \expval{H_0}{\qmstate}$ is independent of $\jarb$, such that Eq.~\eqref{eq:37} leads to (at $\jarb =  \jtot_{\mathrm{stat}}$)
\begin{equation}
- \frac{\jtot_{\mathrm{stat}}}{\rho_0} + \A = \mathbf 0 \iff \mathbf{a}(\qmstate_{\mathrm{m}},\jtot_{\mathrm{stat}}) = \frac{\jtot_{\mathrm{stat}}}{\rho_0} = \A .
\end{equation}

Before concluding, we also take this opportunity to correct an incorrect claim in a previous publication of one of the authors, namely Proposition~8 in Ref.~\onlinecite{Laestadius2014}, which essentially misconstrued the contradiction reached in a reductio ad absurdum proof as a problem with the proof itself. Specifically, Proposition~8 in Ref.~\onlinecite{Laestadius2014} considers two ground state energies $E = E(v,\A)$ and $E' = E(v',\A')$. Leaving aside for the sake of the argument all other problems with Diener's proof idea, one then reaches the contradiction $E+E'<E+E'$, but contrary to Proposition~8 in Ref.~\onlinecite{Laestadius2014} this is not an additional flaw of the attempted proof.

\section{Conclusions}

We have revisited Diener's attempted construction of a density-functional theory featuring the gauge-invariant, total current density. The underlying crucial assumptions have been clarified by a reformulation in terms of a maximin principle. As Diener's construction employs a nonstandard variational principle, it avoids some of the usual difficulties with the total current density as a variational parameter. Nonetheless, we have shown here that his attempted construction fails to establish a current-density-functional theory. Since the correct ground-state energy cannot be obtained within this framework. Moreover, the attempt to establish a Hohenberg--Kohn mapping for total current densities suffers from irreparable gaps in the reasoning. We have shown that there must be counterexamples for which the procedure does not retrieve the correct external vector potential from a given current density.

On the other hand, in broad outline, Diener's formulation shares notable features with the recently proposed Maxwell--Schr\"odinger DFT (MDFT)~\cite{TellgrenSolo2018}, though details differ on crucial points. Diener introduces an effective vector potential, which is equivalent to a total current density, while MDFT takes the induced magnetic field into account that is equivalent to a vector potential or a current density and the total current density then arises naturally as a basic variable. In both cases, the total current density is a variational parameter that is varied independently from the wave function and the external potentials. Moreover, in Diener's approach this variational parameter originates from a nonstandard, and unfortunately mistaken, re-expression of the Schr\"odinger variational principle. In MDFT, it comes from a modified energy minimization principle that simply adds the energy of the induced magnetic field. One can thus view MDFT as a proof of concept for deriving density-functional theories of the total current from modified variational principles.
The very same considerations incorporating a fully quantized electromagnetic field lead to quantum-electrodynamical DFT (QEDFT)~\cite{Ruggenthaler2015}. Such extended density-functional theories form a physically better motivated and theoretically more sound way for a density-functional framework including the total current density.

\section*{Acknowledgements}

A.~L.\ acknowledges support from the Research
Council of Norway (RCN) under CoE Grant Nos. 287906
and 262695 (Hylleraas Centre for Quantum Molecular
Sciences). 
E.~I.~T.\ acknowledges support from RCN under Grant No.~287950 and ERC-STG-2014 Grant No.~639508. 
M.~P.\ acknowledges support by the Erwin Schrödinger Fellowship J
4107-N27 of the FWF (Austrian Science Fund).
The authors are thankful to L.~Garrigue and M.~A.~Csirik for comments and suggestions that greatly improved the manuscript and moreover 
acknowledge the support of the Centre for Advanced Study (CAS) in Oslo, Norway, which funded and hosted the workshop ``Do Electron Current Densities Determine All There Is to Know?'' during 2018.

 \appendix

\section{Convex functions}
\label{appConvexStuff}

We here review basic definitions from convex analysis. A subset $S$ of a vector space is said to be \emph{convex} if $x_1,x_2 \in S$ implies $\lambda x_1 + (1-\lambda) x_2 \in S$, for all $0 \leq \lambda \leq 1$. A function $f(x)$ is said to be \emph{convex} if it is defined on a convex domain and linear interpolation always yields an overestimate,
\begin{equation}
    f(\lambda x_1 + (1-\lambda) x_2) \leq \lambda f(x_1) + (1-\lambda) f(x_2),
\end{equation}
for all $0 \leq \lambda \leq 1$. A function is said to be \emph{concave} if the reverse inequalities hold.

A function of two variables $f(x,y)$ can be convex in each of the arguments separately,
\begin{equation}
    \begin{split}
        f(\lambda x_1 + (1-\lambda) x_2,y) & \leq \lambda f(x_1,y) + (1-\lambda) f(x_2,y),
        \\
    f(x,\lambda y_1 + (1-\lambda) y_2) & \leq \lambda f(x,y_1) + (1-\lambda) f(x,y_2).
    \end{split}
\end{equation}
A stronger property is \emph{joint convexity} in both arguments, 
\begin{equation}
 \begin{split}
    f(\lambda x_1 & + (1-\lambda) x_2, \lambda y_1 + (1-\lambda) y_2) 
         \\
    & \leq \lambda f(x_1,y_1) + (1-\lambda) f(x_2,y_2).
 \end{split}
\end{equation}

The function
\begin{equation}
    g(z) = \inf_{x} \left( \langle x, z \rangle + f(x) \right),
\end{equation}
where $\langle x, z \rangle$ denotes a scalar product (or bilinear pairing) is concave by construction. Changing the infimum to a supremum yields a convex function. This fact is useful in density-functional theory, since the ground-state energy of standard DFT,
\begin{equation}
    E(v) = \inf_{\rho} \left( \int \rho v \rmd \mathbf{r} + F(\rho) \right),
\end{equation}
is of this form. Convexity properties in B-DFT and CDFT are reviewed in Ref.~\onlinecite{TellgrenSolo2018}.

\section{Proof of existence of Diener minimizers}
\label{app:proof}

A Previous work pointed out the issue of variational collapse, i.e., for some $\rho$ we have $\FD(\rho,\jarb) = -\infty$~\cite{Tellgren2012}. This is a serious, but not totally decisive, problem for Diener's approach. Here we show how it can be circumvented at the cost of breaking gauge invariance by introducing a restriction on the kinetic energy density. Under such conditions, we prove the existence of a minimizer of $\FD(\rho,\jarb)$. 

First some preparations follow. We limit ourselves to pure states $\qmstate = \vert\psi\rangle\langle \psi\vert$.
Define the trial set of physical wave functions by 
 \begin{align}
	 \mathcal W_N := \Big\{ \psi \in H_{}^1(\R^{3N},\mathbb C) : \int_{\mathbb R^{3N}} \vert \psi \vert^2 \rmd \x_1 \cdots \rmd\x_N= 1 \Big\}.
 \end{align}
 The set $\mathcal W_N$ is an intuitive choice for the wave functions being minimized over in the pure-state Diener functional. We may thus take
 \begin{equation}
 \FD(\rho,\jarb) = \inf_{\psi\in \mathcal W_N, \psi \mapsto \rho } \ED(\psi,\jarb) 
 \end{equation}
 as the Diener functional in a more detailed setting. (Also note the slight abuse of notation, as we write $\ED(\psi,\jarb)$ instead of the consistent choice $\ED(\vert \psi\rangle\langle \psi\vert,\jarb)$.)  
 
 The functional $\FD$ can be defined on  the space (see Ref.~\onlinecite{LAESTADIUS_JCTC15_4003})
 \begin{equation}
     \begin{split}
         \mathcal R_N := \Big\{ (\rho,\jarb) : \sqrt{\rho} \in H^1(\R^3,\R), \jarb \in L^1\cap L^{3/2}(\R^3,\R^3),& \\
	 \int_{\mathbb R^3} \rho \rmd \x = N, \int_{\mathbb R^3} \frac{\vert \jarb  \vert^2}{\rho} \rmd \x < + \infty \Big\}&.
     \end{split}
 \end{equation}
Then $\psi \in \mathcal W_N$ implies $(\rho_{\psi},\jpara_{\psi}) \in \mathcal R_N$. 
By a feature called \emph{compatibility}~\cite{LAESTADIUS_JCTC15_4003}, the $\jpara$ and $\jtot$ belong to the same space $\mathcal R_N$ and $\jarb$ is an arbitrary current in that space.
 
 For technical reasons, we introduce a further restricted wave-function space. 
 Let the kinetic energy density of a state $\psi$, $\tau_\psi: \R^3 \to [0,+\infty]$, be given by 
 \begin{equation}
 \tau_\psi(\x_1) = \int_{\mathbb R^{3(N-1)}} \tfrac{1}{2} \vert \nabla \psi\vert^2 \rmd \x_2 \cdots \rmd\x_N  .
 \end{equation}
 For $g\geq 0$ a fixed integrable function with $\int g \rmd \x =C$, $C>0$, set
 \begin{align}
	 \mathcal W_N^C := \big\{ \psi \in \mathcal W_N : \tau_\psi\leq g, \tau_\psi/\rho_\psi \in L^\infty  \big\} 	 .
 \end{align}
Note that the kinetic energy $\int \tau_\psi \rmd\x$ is bounded above by $C$ for all $\psi \in \mathcal W_N^C$, so we get a kinetic-energy cutoff for all such wave functions. 
The constraint $\tau_\psi/\rho_\psi \in L^\infty$ is to guarantee $\jpara_\psi/\rho_\psi \in L^\infty$, a property used later. 
In the mathematical literature of CDFT $\psi$ is typically assumed to have finite kinetic energy (here we go further and have even imposed a pointwise bound on the kinetic-energy density $\tau_\psi$).

The following lemma is an adaptation of Proposition~5 in Ref.~\onlinecite{Laestadius2014}. 

\begin{lemma}\label{lemma:min}
Fix $\rho$ and $\jarb$.  
Suppose weak convergence of $\psi^{n}$ to $\psi_\mathrm{m}$ in $H^1$ and that $\rho_{\psi^n} = \rho$ and $\jarb= \jpara_{\psi^n} +  \mathbf a(\psi^n,\jarb) \rho$ hold for all $n$. Then $\rho_{\psi_\mathrm{m}} = \rho$ and, for a subsequence $\{n_k\}$, $\jpara_{\psi_\mathrm{m}}$ is the weak $L^1$ limit of the $\{ \psi^{n_k}\}$'s paramagnetic current densities as well as their pointwise limit (a.e.),
\begin{equation}\label{eq:lem-jp-limits}
\jpara_{\psi^{n_k}} \rightharpoonup  \jpara_{\psi_\mathrm{m}} \quad (\mathrm{in}\,\,L^1),
\quad \jpara_{\psi^{n_k}} \to \jpara_{\psi_\mathrm{m}} \quad \mathrm{(a.e.)}
\end{equation}
\end{lemma}
\begin{proof}
By Theorem~3.3 in Lieb~\cite{Lieb1983} we know that for a subsequence, indexed by $n_k$, $\psi^{n_k} \to \psi_\mathrm{m}$ in $L^2$ and that $\rho_{\psi_\mathrm{m}} = \rho$. 
Next we demonstrate that $\jpara_{\psi^{n_k}} \rightharpoonup  \jpara_{\psi_\mathrm{m}}$ weakly in $L^1$ and pointwise (a.e.). Set
$J(\psi)$ to be any component of
\begin{equation}
 \int_{\mathbb R^{3(N-1)}} \overline{\psi} \nabla \psi \rmd \x_2 \cdots \rmd \x_N
\end{equation}
and note, for $\phi$ being the characteristic function of any measurable set $M\subset \mathbb R^3$, that ($j$ indexing the component)
\begin{equation}
    \begin{split}
        &\int_{\mathbb R^3}  \phi\, ( J(\psi^n) - J(\psi_\mathrm{m}) ) \rmd\x_1  \\
    & = \int_{\mathbb{R}^3} \phi  \int_{\mathbb R^{3(N-1)}} 
    [(\overline{\psi}^n \nabla \psi^n - \overline{\psi}_\mathrm{m} \nabla \psi_\mathrm{m} )]_j \rmd\x_2 \cdots \rmd \x_N  \\
    &=\int_{\mathbb R^{3N} } \phi\,\Big[ ( \overline{\psi}^n -\overline{\psi}_\mathrm{m}) \nabla \psi^n
    - \overline{\psi}_\mathrm{m} \nabla (\psi_\mathrm{m} - \psi^n) \Big]_j  \rmd \x_1 \cdots \rmd \x_N  \\
    &= \int_{\mathbb R^{3N} } \phi\, (\overline{\psi}^n -\overline{\psi}_\mathrm{m}) [\nabla \psi^n ]_j\rmd \x_1 \cdots \rmd \x_N \\
    &\quad - \int_{\mathbb R^{3N} } \phi \,\overline{\psi}_\mathrm{m} [ \nabla (\psi_\mathrm{m} - \psi^n) ]_j\rmd \x_1 \cdots \rmd \x_N \\
    &=: A_j(n) - B_j(n)   .
    \end{split}
\end{equation}
Since $ \phi \nabla \psi^n  \in L^2$ we have $A_j(n_k)\to 0$ as $k\to\infty$ by norm convergence of $\{ \psi^n\}$, and from $ \phi \psi\in L^2$ we obtain $B_j(n_k) \to 0$ by weak convergence in $L^2$ of $\{ \nabla \psi^n \}$. 
We can also use this argument (take $\phi\in L^\infty$) to obtain weak convergence in $L^1$ of (a subsequence of) $\jpara_{\psi^n}$ to $\jpara_{\psi_\mathrm{m}}$.
\end{proof}

Now, define the kinetic-cutoff Diener functional $\FD^C(\rho,\jarb)$ by
\begin{equation}\label{eq:appD}
    \FD^C(\rho,\jarb) =  \inf_{\psi\in \mathcal W_N^\mathrm{C}, \rho_\psi = \rho } \ED(\psi,\jarb).
\end{equation}
For this version of the Diener functional we can finally establish existence of Diener minimizers. 

\begin{proposition} \label{prop:min}
Let $(\rho,\jarb) \in \mathcal R_N \cap \{(\rho,\jarb) : \jarb/\rho \in L^\infty \}$. 
There exists $\psi_\mathrm{m} \in \mathcal W_N^C$ such that $\FD^C(\rho,\jarb) = \ED(\psi_\mathrm{m},\jarb)$ and $\psi_\mathrm{m} \mapsto \rho$, i.e., the infimum in Eq.~\eqref{eq:appD} is a minimum. 
\end{proposition}

\begin{proof} 

Let $\psi^n$ be a minimizing sequence in $\mathcal W_N^C$, i.e., 
\begin{equation}
\mathcal E(\psi^n,\jarb) \to \FD^C(\rho,\jarb), \quad \rho_{\psi^n}= \rho,\quad \jarb= \jpara_{\psi^n} +  \mathbf a(\psi^n,\jarb) \rho.  
\end{equation}
This implies that
$\Vert \psi^n \Vert_{H^1}^2 = 1 + 2\int \tau_{\psi^n} \rmd\x \leq 1 +2C$
and thus by the Banach--Alaoglu theorem there exists a weakly convergent subsequence and an element $\psi_\mathrm{m}\in H^1$ such that $\psi^{n_k} \rightharpoonup \psi_\mathrm{m}$ (weakly in $H^1$). 
By Lemma~\ref{lemma:min}, the limit function $\psi_\mathrm{m}$ has the particle density $\rho$, and $\jpara_{\psi_\mathrm{m}}$ is the weak $L^1$ limit of the $\{ \psi^{n_k}\}$'s paramagnetic current densities as well as the pointwise limit (a.e.) like given in Eq.~\eqref{eq:lem-jp-limits}.
Then with
\begin{align}
    \mathbf a(\psi_\mathrm{m},\jarb) = \frac{\jarb- \jpara_{\psi_\mathrm{m}}}{\rho},
\end{align}
$\psi_\mathrm{m}$ can be taken as a candidate for a minimizer, where by definition $\mathcal E(\psi_\mathrm{m}, \jarb) \geq  F_\mathrm{D}^C(\rho,\jarb)$.
What remains to be verified is the reverse inequality.
To meet that end, set
\begin{align}
    \mathbf a^n := \mathbf a(\psi^n,\jarb) = \frac{\jarb - \jpara_{\psi^n}}{\rho} ,
\end{align}
which is an element of $L^\infty$ since $\jarb/\rho \in L^\infty$ and since for $\psi\in \mathcal W_N^C$
\begin{equation}
    \begin{split}
        \frac{\jpara_\psi}{\rho_\psi} \leq \left(\frac{\tau_\psi}{\rho_\psi} \right)^{1/2} \in L^\infty .
    \end{split}
\end{equation}
Then, pointwise a.e.\ and weakly in $L^1$ we have for a subsequence
\begin{align}
    \jarb = \lim_k (\jpara_{\psi^{n_k}} + \mathbf a^{n_k} \rho) = \jpara_{\psi_\mathrm{m}} +  ( \lim_k \mathbf a^{n_k}) \rho .
\end{align}
This gives that $\mathbf a^{n_k} \rho \rightharpoonup  \mathbf a(\psi_\mathrm{m},\jarb) \rho$ weakly in $L^1$ as well as pointwise a.e.\ (for a subsequence). Moreover, using again the fact that $\jarb/\rho \in L^\infty$ and Lemma~\ref{lemma:min}, we obtain
by dominated convergence since $ \tfrac{1}{2}\vert \jpara_{\psi^{n_k}} \vert^2 \rho^{-1} \leq  \tau_{\psi^{n_k}} \leq g \in L^1$ that
\begin{equation}\label{eq:step1}
\begin{split}
&\lim_k \int \rho \vert \mathbf a^{n_k} \vert^2 \rmd \x   \\
& = \int \frac{\vert \jarb \vert^2}{\rho} \rmd \x -2 \lim_k \int\frac{\jarb\cdot \jpara_{\psi^{n_k}} }{\rho} \rmd \x + \lim_k  \int  \frac{ \vert \jpara_{\psi^{n_k}} \vert^2 }{\rho} \rmd\x \\
& =  \int \frac{\vert \jarb- \jpara_{\psi_\mathrm{m}} \vert^2 }{\rho} \rmd\x= \int\rho \vert \mathbf a(\psi_\mathrm{m},\jarb) \vert^2 \rmd \x.
\end{split}
\end{equation}
Consequently, 
\begin{equation}
    \begin{split}
 \ED(\psi_\mathrm{m},\jarb) 
    &=\langle \psi_\mathrm{m}, H_0\psi_\mathrm{m}\rangle -   \int\tfrac 1 2 \rho \vert \mathbf a(\psi_\mathrm{m},\jarb) \vert^2\rmd \x \\
    &\leq \lim_k \left( \langle \psi^{n_k}, H_0 \psi^{n_k}\rangle   -  \int \tfrac 1 2 \rho \vert \mathbf a^{n_k} \vert^2   \rmd\x \right) \\
    &= F_\mathrm{D}^C(\rho,\jarb)
        \end{split}
\end{equation}
follows by Eq.~\eqref{eq:step1} above and lower semi-continuity of the quadratic form $\psi\mapsto \langle \psi,H_0 \psi \rangle$. 
\end{proof}

\bibliography{refs}

\end{document}